\documentclass[a4paper,12pt,reqno,superscriptaddress,nofootinbib]{revtex4}
\usepackage[centertags]{amsmath}
\usepackage{amsfonts}
\usepackage{amssymb}
\usepackage{amsthm}
\usepackage{newlfont}
\usepackage{stmaryrd}
\usepackage{mathrsfs}
\usepackage{mathtools}
\usepackage{euscript}
\usepackage{graphicx}
\usepackage{enumerate}
\usepackage{todonotes}
\usepackage[normalem]{ulem} % for strikeout text with \sout

\usepackage{color}
\usepackage{floatrow}
\usepackage{caption}

\usepackage{tikz}
\usepackage{pgf}
\usetikzlibrary{positioning,fit,calc}
\usetikzlibrary{arrows,automata}
\usepackage{wrapfig}
\usepackage{subfigure}
\usepackage{amscd}
\usepackage{hyperref}

\theoremstyle{plain}
  \newtheorem{theorem}{Theorem}[section]
  
  \newtheorem{proposition}[theorem]{Proposition}
  
\theoremstyle{definition}

\theoremstyle{remark}

\newcommand{\opunit}{\text{1}\kern-0.22em\text{l}}

\DeclareMathAlphabet{\mathpzc}{OT1}{pzc}{m}{it}

\let\oldsqrt\sqrt
\def\sqrt{\mathpalette\DHLhksqrt}
\def\DHLhksqrt#1#2{%
\setbox0=\hbox{$#1\oldsqrt{#2\,}$}\dimen0=\ht0
\advance\dimen0-0.2\ht0
\setbox2=\hbox{\vrule height\ht0 depth -\dimen0}%
{\box0\lower0.4pt\box2}}

\begin{document}

\title{\textbf{Glassy states: the free Ising model on a tree }}

\author{Daniel Gandolfo}
\affiliation{Aix Marseille Univ, Universit\'{e} de Toulon, CNRS, CPT, Marseille, France}
\author{Christian Maes}
\affiliation{Instituut voor Theoretische Fysica, KU Leuven}
\author{Jean Ruiz}
\affiliation{Aix Marseille Univ, Universit\'{e} de Toulon, CNRS, CPT, Marseille, France}
\author{Senya Shlosman}
\affiliation{Aix Marseille Univ, Universit\'{e} de Toulon, CNRS, CPT, Marseille, France}
\affiliation{Skolkovo Institute of Science and Technology}
\affiliation{Institute of
Information Transmission Problems, RAS, Moscow}

\begin{abstract}
We consider the ferromagnetic Ising model on the Cayley tree and we investigate the
decomposition of the free state into extremal states below the spin glass
temperature. We show that this decomposition has uncountably many components.
The tail observable showing that the free state is not extremal
 is related to the Edwards-Anderson parameter, measuring the variance of the (random) magnetization obtained from
drawing boundary conditions from the free state.
\end{abstract}

\maketitle

%\tableofcontents

\section{Introduction}

It is well known that the Ising model on a regular Cayley tree undergoes a  {second} order phase transition at the
critical temperature $T_{\operatorname{cr}}$, below which the Gibbs states $\mu_{+}$ and $\mu_{-}$, corresponding to
${+}$ and ${-}$ boundary conditions, are different extremal states. Unlike the usual $\mathbb{Z}^{d}$-lattice case, on the tree the behavior of the free state $\mu_{\varnothing}$, corresponding to
empty boundary conditions, is very rich. On $\mathbb{Z}^{d}$ we have
$\mu_{\varnothing}=\frac{1}{2}(\mu_{+}+\mu_{-})$, while on the tree that is  {trivially} true only in the uniqueness regime. Moreover, the free state is extremal for temperatures $T$ below  the
critical temperature $T_{\operatorname{cr}}$, until a certain \textit{spin-glass} temperature $T_{\operatorname{SG}}$, below which it stops to be  {extremal}.
The question of finding the temperature range of the ergodicity of the state
$\mu_{\varnothing}$ was open for twenty years, and was solved by Bleher, Ruiz
and Zagrebnov  in their 1995 paper  \cite{BRZ}. Soon after, a simpler argument
was provided by Ioffe, \cite{Io}. For a closely related communication theory problem,
see the "Broadcasting on trees" paper by Evans,  Kenyon, Peres and
Schulman \cite{EKPS}.

In the present paper we study the free state at temperatures below $T_{\operatorname{SG}}$.

A principal static feature of the spin glass phase is the
presence of infinitely many pure states; see the discussion in \cite{NS} and the references therein.  By a gauge transformation the spin glass on the tree and the ferromagnet are equivalent, except that random boundary conditions for the ferromagnet correspond to fixed boundary conditions for the spin-glass, as was for example discussed in \cite{CCST}. We show that the same phenomenon  {also} happens
in the ferromagnetic Ising model on the tree, {\it i.e.}, without randomness
in the interaction. Namely,  the free
state of the Ising model  below the spin-glass
temperature has a decomposition into extremal states which involves
uncountably many extremal states. That is why we call this state `glassy'. Our
result answers an older question of Arnout van Enter, \cite{VE}.

The next section contains the more detailed question with some notations and definitions. Our main decomposition result is given in Section 3. The remaining sections contain
further details and proofs.

\section{Notations and definitions}
Let ${\cal T}_k = (V, E)$ be the Cayley tree with branching ratio $k\geq 2.$
We consider the nearest neighbor Ising model, where spins $\sigma_x=\pm 1$ have a Gibbs distribution  at
temperature $T = 1 / \beta$ with boundary condition $\eta$ in a finite volume $\Lambda$ given by
\begin{equation}
  \label{IGM} {\mu} (\sigma) = Z^{- 1} \exp \left\{ \beta \sum_{\langle x,
  y \rangle} J_{xy} \sigma_x \sigma_y + \beta \sum_{\langle x, y \rangle}
  J_{xy} \sigma_x \eta_y \right\}
\end{equation}
Both sums run over nearest neighbors pairs, the first being over the pairs
$x \in \Lambda$,
$y \in \Lambda$, and the second one runs over sites $x \in \Lambda, y \notin
\Lambda$.  The infinite-volume Gibbs distributions are obtained as  {the convex hull of the set of all possible limit point} when $\Lambda$ grows to cover all the vertices of the tree.
The set of Gibbs distributions for a fixed temperature and interaction is convex, and its extreme points are called pure states: they cannot be decomposed into other states.
%Speaking informally, these states are giving precise values to macroscopic observables and in that sense correspond to some equilibrium condition.
\\

In the ferromagnetic case we put $J_{xy} = 1$, while in the spin-glass model the
interaction is random: $J_{xy} = \pm 1$ with probability $1 / 2$ independently
for any pair $\langle x, y \rangle$.  We can also consider the random interaction
\begin{equation}
  \label{RandCoup} J_{xy} = \left\{ \begin{array}{l}
    - 1 \hspace{1em}  \hspace{1em} \operatorname{with} \operatorname{probability} \hspace{1em}
    p\\
    + 1 \hspace{1em}  \hspace{1em} \operatorname{with} \operatorname{probability} \hspace{1em}
    1 - p
  \end{array} \right.
\end{equation}
which interpolates between the ferromagnetic model $p = 0$ and the spin-glass
model $p = 1 / 2$.\\

It is well known that in the ferromagnetic case the phase transition happens at  the critical temperature
$T_{\operatorname{cr}} = 1 / \operatorname{arctanh} (1 / k)$. For $T>T_\text{cr}$ there is a unique infinite-volume Gibbs distribution and below that critical temperature the spontaneous magnetization $m^*(T)$ becomes nonzero, while
%there is spontaneous magnetization in the sense that, always for the ferromagnetic model,
plus and minus boundary conditions give rise to different states:
\[
\langle \sigma_0 \rangle^T_{\pm} = \pm m^*(T), \qquad m^*(T)>0 \text{ iff } T<T_\text{cr}.
\]
We use the convention that % the nature of
the boundary condition  by which the infinite-volume Gibbs state is obtained is indicated by a subscript to the expectation $\langle \cdot\rangle$, and superscripts will be used to indicate further parameters like temperature or the choice of model in \eqref{RandCoup}.\\
There is yet another special temperature, $T_{\operatorname{SG}} = 1 /
    \operatorname{arctanh} \left( 1 / \sqrt{k} \right)$, called the spin-glass temperature.  It will often appear in what follows below, and it has various interpretations.  For the ferromagnetic model it is known that the free state, the infinite-volume Gibbs  distribution obtained by putting $\eta\equiv 0$ in \eqref{IGM}, is extreme for $T > T_\text{SG}$ while it is not for $T< T_\text{SG}$;
    see \cite{BRZ},\cite{Io},\cite{EKPS},\cite{MSW}.
    Hence the question that motivates the present paper: what is the decomposition
of the free state $ \langle\cdot\rangle^T_{\varnothing} $ (at these lower temperatures) into extremal ones?\\% \mu_T_{\varnothing}\emptyset
Part of the question is also to understand why that extremality of the free state exactly stops at the spin-glass temperature, which in its origin characterizes a transition to glassy behavior in the spin-glass model.  For example consider the spin-glass state obtained via plus-boundary conditions and look at the (random)  {local} magnetization  {at the origin}:
\[
\langle \sigma_0\rangle^{T,\text{SG}}_+ = m(T, \{J_{xy}\})
\]
It depends on the independent random variables $J_{xy}$, which take values $\pm 1$ with equal probability.  The spin-glass temperature $T_{\operatorname{SG}}$ is that temperature below which $m(T, \{J_{xy}\})$ is a fluctuating quantity with a non-trivial distribution.  For example,{ the second moment, the Edwards--Anderson parameter, is positive only there:}
\[
q_{EA}(T) = \mathbb{E} [m^2(T, \{J_{xy}\})] >0 \qquad \text{ iff } T< T_\text{SG},
\]
where the expectation $\mathbb{E}[\cdot]$ is taken over the randomness $\{J_{xy}\}$.

\section{Decomposition of the free state}

%The question was asked by Aernout.
In this section we present the decomposition of the free state into pure states and we explain that a continuum of them enters into it, at least
at low temperatures. Presumably, it is the case for all temperatures in the spin-glass region.

\subsection{The construction of the decomposition}

For any Gibbs state $\mu$ corresponding to
temperature $T$ we have that%, quite informally.
\[
\mu\left(  \cdot\right)  =\int \mu(d\sigma)\,[\left\langle \cdot\right\rangle^{T}_\sigma]
\]
where $\sigma$ is a spin configuration drawn from $\mu$ and used as boundary condition at infinity.  The Gibbs distributions $\left\langle \cdot\right\rangle_\sigma^T$ are
 $\mu$-almost surely extreme. Hence, we have here a decomposition of $\mu$ into extreme Gibbs distributions, but obviously the states
 $\left\langle
 \cdot
\right\rangle_\sigma^T$ might well be the same for
different $\sigma$-s. Nevertheless this decomposition is nontrivial when $\mu$ is not extremal. % That is the strategy we apply for the free state on the tree.
It remains then to see how many and which different extremal states we get.\\

For the Ising model free state $\mu^T_{\varnothing}$ we will use the
Edwards-Sokal representation, \cite{mono}:%
\begin{equation}
\mu^T_{\varnothing}\left(  \cdot\right)  =\mathbb{E}_{q(T)}\left(
\mu^{ES}\left(  \cdot{\Huge |}n\right)  \right)   \label{01}%
\end{equation}
(In contrast to the more standard notation we prefer here to call $q(T) = 1-\tanh 1/T$ the
probability of removed (or closed) bonds.) On the tree, the random
cluster measure is generated by independent bond percolation and $n$ is
the resulting random bond configuration over which we take the expectation $\mathbb{E}_{q}$. The open bonds generate a partition
of the tree into maximal connected components. The measure $\mu^{ES}\left(
d\sigma{\Huge |}n\right)  $ is supported by the spin configurations $\sigma$ which
are constant on each connected component as specified by $n;$ these
constants take values $\pm 1$, independently with probability 1/2.\\

We can rewrite \eqref{01} on the tree by ordering the components using the following definitions.\\
Let $D$ be
a subset of bonds of the tree ${\cal T}_k$, and  consider the two Ising
spin configurations $\sigma^{D, \pm}$ \ on ${\cal T}_k$ defined as:
\begin{equation}\label{dspin}
      \sigma_0^{D, +}  =  + 1, \hspace{1em}
    \sigma_0^{D,_-} = - 1
\end{equation}
and
\begin{eqnarray*}
  \sigma_x^{D, \pm} & = & - \sigma_y^{D, \pm} \text{ \ \ \ $\text{for } (x, y)
  \in D$ }
  \\
  \sigma_x^{D, \pm} & = & + \sigma_y^{D, \pm} \text{ \ \ \ for } (x, y) \notin D
\end{eqnarray*}
That is, we fix the value of the spin at the root (say $0$) to be $+ 1$ or $- 1$,
and the nearest neighbor spins alternate iff the corresponding bond
belongs to the set $D$.

%%%%%%%%%%%%%%%%%%%%%%%%%%%%%%%%%%%%%%%%%%%%%%%%%%%%%%%%%%%%%%%
By $ \langle \cdot \rangle^T_{\sigma^{D, \omega}}$
we denote the Gibbs state of the ferromagnetic Ising model at inverse
temperature $T$ with the boundary condition $\sigma^{D, \omega}$ where $\omega
= \pm 1$ corresponds to the way the spin at the origin is chosen.

Let $p \in (0, 1)$. Take the set $D$ to be random: every bond
decides to be in $D$ with probability $p$ independently of the other bonds.
 Denote by $\mathbb{E}_p$
 the expectation with respect to that process.

\begin{proposition}
The following decomposition of the free state for the ferromagnetic Ising model on the tree holds for all temperatures $T$:
\begin{equation}\label{es}
 \langle\cdot\rangle^T_{\varnothing} = \frac 1{2} \mathbb{E}_{p(T)}
\left[ \langle \cdot\rangle^T_{\sigma^{D,+}} + \langle \cdot\rangle^T_{\sigma^{D,-}} \right],
 \end{equation}
 where
 \[
 p(T) = \frac 1{2}[ 1 - \tanh 1/T].
 \]
  \end{proposition}

  \begin{proof}
  We apply the Edwards-Sokal representation \eqref{01} of the Ising model. Start by noting that
  \[
 \left\langle \cdot\right\rangle^T_{\varnothing} = \mathbb{E}_{q(T)}\left(
 \mu^{ES}\left(  \cdot{\Huge |}n\right)  \right)  =\mathbb{E}_{q(T)}\int  \left\langle
\cdot\right\rangle^T_{\sigma}\mu^{ES}(d\sigma{\Huge |}n),
\]

Consider now for a given bond collection $n$ the atomic measure
 $\mu^{\pm}\left(  d\sigma{\Huge |}n\right)  =\frac{1}%
{2}\left(  \delta_{\sigma^{n+}}+\delta_{\sigma^{n-}}\right)$, which to $n$ assigns two configurations defined by the relations  \eqref{dspin}, each with probability
$\frac{1}{2}.$  In other words, fixing the spin of the origin and fixing $n$ determines all the other spin values, where in particular neighboring connected components of open bonds alternate their spin.
However the resulting spin configuration would have the same distribution as in the Edwards-Sokal representation with twice as large probability of closed bonds: with $p(T) = q(T)/2$,
\[
\mathbb{E}_{q(T)}[\mu^{ES}\left(  d\sigma{\Huge |}n\right)]    = \mathbb{E}_{p(T)}[\mu^{\pm
}\left(  d\sigma{\Huge |}n\right)],
\]
%where $p(T) = q(T)/2$,
%and where $\mu^{\pm}\left(  d\sigma{\Huge |}n\right)  =\frac{1}%
%{2}\left(  \delta_{\sigma^{n+}}+\delta_{\sigma^{n-}}\right)$ is atomic, {\it i.e.}, it
%assigns to any bond collection $n$ two configurations defined by the relations  \eqref{dspin}, each with probability
%$\frac{1}{2}.$
Bringing all that together we conclude that on the tree
\eqref{01} reduces to the formula \eqref{es}.
 \end{proof}

Note also that the states $\langle \cdot\rangle^T_{\sigma^{D,+}}$, $\langle \cdot\rangle^T_{\sigma^{D,-}}$
%\end{document}
can be obtained as thermodynamic limits of the finite-volume Gibbs states with the boundary conditions $\sigma^{D,+}$, $\sigma^{D,-}$.
These limits exist for $ \langle\cdot\rangle^T_{\varnothing}$-almost all boundary conditions $\sigma^{D,+}$, $\sigma^{D,-}$.

  To see that this decomposition \eqref{es} is non-trivial  for $T<T_\text{SG}$ it suffices to show that when $T< T_\text{SG}$,
$\langle \cdot\rangle^T_{\sigma^{D,+}} \neq  \langle \cdot\rangle^T_{\sigma^{D,-}}$ for $\mathbb{E}_{p(T)}$-typical sets $D$.
That follows from the relations
\begin{equation}\label{posor}
\mathbb{E}_{p(T)}[ \langle \cdot\rangle^T_{\sigma^{D,+}}] > 0 > \mathbb{E}_{p(T)}[ \langle \cdot\rangle^T_{\sigma^{D,-}}],
\end{equation}
which are a special case of the Theorem 1.1 of \cite{EKPS}.

Of course, the states $\langle\cdot\rangle_{\sigma^{D,+}}^{T}$ may coincide
for different sets $D$ -- this is the case at high temperature. In the
next subsection we will show that at low temperatures there is a continuum of
different states $\langle\cdot\rangle_{\sigma^{D,+}}^{T}$ as we vary over the
sets $D$, see also \cite{mani}. There we have shown that if the set $D$
consists of bonds sufficiently separated from one another, then the
configuration $\sigma^{D,+}$ is a stable ground state. Of course, for our
random configuration $D$ this is not the case; for example we will see in $D$
pairs of bonds sharing a vertex, which will happen with positive density.
However, \textit{quite often} we will see just isolated bonds, well separated,
once $p(T)$ is small enough, see the next subsection for more details.

\textbf{Remark.} At all temperatures $T$ the free-state two-point function
$\langle\sigma_{0}\sigma_{x}\rangle_{\varnothing}^{T}\rightarrow0$ goes to zero
when $x$ goes to infinity, yet for $T<T_{\text{SG}}$ the free state is not
extreme. In particular, the magnetization in increasingly large volumes has a
variance that does not go to zero with the size of the volume. The interesting
tail-observable which shows that the free state is not extreme is related to
the Edwards-Anderson parameter. Here is the simplest version: take the
magnetization  {at the origin,}
\[
M(\tau)=\langle\sigma_{0}\rangle_{\tau}^{T}%
\]
in the infinite-volume Gibbs distribution with boundary condition $\tau$; that
$\tau$ is drawn from the free state at temperature $T$. For $T<T_{\text{SG}}$
the random variable $M(\tau)$ has a non-trivial distribution.

\subsection{At low temperatures \texorpdfstring{$T$}{Lg} the states \texorpdfstring{$\langle\cdot\rangle
_{\sigma^{D,+}}^{T}$}{Lg} are mutually singular.}

\begin{theorem}
Pick two independent configurations $\sigma^{D_1,+}$, $\sigma^{D_2,+}$ from the distribution
 $\langle\cdot\rangle^T_{\varnothing}$. Then the two limiting states \texorpdfstring{$\langle\cdot\rangle
_{\sigma^{D_1,+}}^{T}$}{Lg}, \texorpdfstring{$\langle\cdot\rangle
_{\sigma^{D_2,+}}^{T}$}{Lg} exist and are mutually singular a.s. with respect to
$\langle\cdot\rangle^T_{\varnothing}\times \langle\cdot\rangle^T_{\varnothing}$,
provided the temperature $T$ is low enough.
\end{theorem}

\textbf{Proof.} We denote by $D$ a random
configuration of bonds in $\mathcal{T}_{k},$ each bond picked independently
with probability $p,$ with the parameter $p$ being fixed and small enough. We
will study the Gibbs states $\langle\cdot\rangle_{\sigma^{D,+}}^{T}$ at low
temperatures, with the goal of showing their a.s. mutual singularity.

Let $b=\left(  x,y\right)  \in\mathcal{T}_{k}$ be a bond, and $B_{N}\left(
b\right)  \subset\mathcal{T}_{k}$ be a ball of radius $N$ centered at $b.$
Consider the ground state (=zero temperature Gibbs) measure $\langle
\cdot\rangle_{\sigma^{D,+}}^{T=0,B_{N}\left(  b\right)  }$ in the box
$B_{N}\left(  b\right)  $ with boundary condition $\sigma^{D,+}.$ We call the
bond $b$ frustrated in the state $\langle\cdot\rangle_{\sigma^{D,+}%
}^{T=0,B_{N}\left(  b\right)  },$ if the event $\sigma_{x}\sigma_{y}=-1$
happens with probability one in the state $\langle\cdot\rangle_{\sigma^{D,+}%
}^{T=0,B_{N}\left(  b\right)  },$ for all $N\ $large enough. We call the bond
$b$ to be $r$-strongly-frustrated (or just $r$-frustrated) in the state
$\langle\cdot\rangle_{\sigma^{D,+}}^{T=0,B_{N}\left(  b\right)  },$ if the
event $\sigma_{x}\sigma_{y}=-1$ happens with probability one in the state
$\langle\cdot\rangle_{\sigma^{D,+}}^{T=0,B_{N}\left(  b\right)  },$ as well as
the events $\sigma_{x^{\prime}}\sigma_{y^{\prime}}=1$ for all bonds
$b^{\prime}=\left(  x^{\prime},y^{\prime}\right)  $ within distance $r$ from
the bond $b,$ again for all $N\ $large enough.

For example, the above will hold if $b\in D,$ while $D$ is a deterministic
configuration composed from \textit{isolated} bonds which are sufficiently far
away from each other, see \cite{mani}, \cite{lob}. What we want to show now is
that if $D$ is \textit{random,} and $b\in D,$ then it is very likely that $b$
is $r$-frustrated, provided $p$ is small enough (depending on $r$). Once we
show that, our claim about mutual singularity will be proven, since for two
independent configurations $D^{\prime},D^{\prime\prime}$ we will be able to
find arbitrarily large \textit{disjoint} sets of $r$-frustrated bonds.

So let $D$ be random, and $b\in D.$ Our first observation is that the
probability of $D$ having other bonds at distance $2r$ from $b$ is quite
small, provided $p$ is small enough. That would be the end of the story if the
configuration $\sigma^{D,+}$ would be a ground state configuration. Indeed, in
that case the state $\langle\cdot\rangle_{\sigma^{D,+}}^{T,B_{N}\left(
b\right)  }$ would be a small perturbation of the configuration $\sigma^{D,+}$
once $T$ is low, uniformly in $N.$

However, the configuration $\sigma^{D,+}$ is not a ground state configuration
a.s., so the state $\langle\cdot\rangle_{\sigma^{D,+}}^{T=0,B_{2r}\left(
b\right)  }$ might have other frustrated bonds in $B_{2r}\left(  b\right)  ;$
moreover, it even can happen that $b$ itself is \textit{not }frustrated in
this state. We will show now that all this is highly unlikely, once $p$ is
small enough.

So suppose the set of frustrated bonds of the state $\langle\cdot
\rangle_{\sigma^{D,+}}^{T=0,B_{2r}\left(  b\right)  }$ is not the singleton
$\left\{  b\right\}  .$ That can happen iff there is a contour $\gamma,$
$\left[  \gamma\cup\mathrm{Int}\left(  \gamma\right)  \right]  \cap
B_{2r}\left(  b\right)  \neq\varnothing,$ crossing certain number $\ell\geq k+1$
of bonds of $\mathcal{T}_{k},$ such that $\left\vert \gamma\cap D\right\vert
\geq\frac{\ell}{2}.$ Consider the set $T_{\gamma}$ of the bonds of $\mathcal{T}%
_{k}$ which are inside $\gamma,$ and the set $L_{\gamma}$ of bonds of
$\mathcal{T}_{k}$ the contour $\gamma$ is intersecting. Together they form a
finite tree $S_{\gamma},$ which has the same branching number $k$ as our
infinite tree $\mathcal{T}_{k}.$ The set $L_{\gamma}$ is the set of all leaves
of the tree $S_{\gamma}.$ Let $n_{\gamma}$ be the number of nodes of
$S_{\gamma}$ inside $\gamma,$ and $\bar{L}_{\gamma}\subset L_{\gamma}$ be the
intersection $L_{\gamma}\cap D.$ So we have $\left\vert T_{\gamma}\right\vert
=n_{\gamma}-1,$ $\left\vert L_{\gamma}\right\vert =\ell,$ $\left\vert \bar
{L}_{\gamma}\right\vert \geq\frac{\ell}{2},$ and $\left\vert L_{\gamma
}\right\vert =1+n_{\gamma}\left(  k-1\right)  .$

Evidently, the probability of seeing such a tree%
\[
\mathbf{\Pr}\left(  S_{\gamma},L_{\gamma},\bar{L}_{\gamma}\right)  \leq
p^{\ell/2},
\]
so%
\[
\mathbf{\Pr}\left(  S_{\gamma},L_{\gamma}\right)  \leq2^{\ell}p^{\ell/2}.
\]
The number of trees $S$ with $n$ inner nodes does not exceed $k^{2n}.$ Thus
the probability that a given bond $b_{1}$ is a leaf of such a tree with $\ell$
leaves is bounded from above by%
\[
2^{\ell}k^{2\ell/\left(  k-1\right)  }p^{\ell/2},
\]
which is exponentially small in $\ell$ for $p$ small enough. So we can apply the standard Peierls
argument to shows that the probability of the event%
\[
\{\text{there is a contour }\gamma,\text{ such that }\gamma\cap B_{2r}\left(
b\right)  \neq\varnothing\text{ or }B_{2r}\left(  b\right)  \subset
\mathrm{Int}\left(  \gamma\right)  \}
\]
goes to zero as $p\rightarrow0,$ which ends the proof.

To conclude, we point out for clarity that the ``Gibbs ground'' states
$\lim_{T\rightarrow0}\langle\cdot\rangle_{\sigma^{D,+}}^{T}$ constructed from
the $p$-random bond configurations $D$, are typically nontrivial measures,
i.e. they have infinite supports, a.s. This is in contrast with the ground
states constructed in \cite{mani}, where the corresponding Gibbs ground states
are supported by a single ground state configuration. However, as is explained
above, a vast majority of the frustrated bonds under typical ground state
$\langle\cdot\rangle_{\sigma^{D,+}}^{T=0}$ are isolated bonds, once $p$ is
small. This fact is the source for the decomposition \eqref{es} to have a
continuum of extremal components.

\section{Double--temperature Ising model}
{We already mentioned in the introduction that the phase diagram of the ferromagnetic Ising model is essentially determined by the critical temperature $T_{\operatorname{cr}} = 1 / \operatorname{arctanh} (1 / k)$, and the spin-glass temperature $T_{\operatorname{SG}} = 1 /
    \operatorname{arctanh} \left( 1 / \sqrt{k} \right)$.  A clarification of the situation can however be obtained by enlarging} the objects in \eqref{es} into a two-temperature setting.  We consider two-temperature states, with $T_2$ the bulk temperature and $T_1$ the boundary temperature,
\[
\nu(T_1,T_2) := \langle \cdot \rangle^{T_2}_{\sigma^{D(T_1),+}}
\]
which %, for the pair $(T_1,T_2)$
is the infinite-volume Gibbs distribution  at temperature $T_2$  with the boundary condition  taken to be the spin configuration \eqref{dspin} where $D$ is drawn from $\mathbb{E}_{p(T_1)}$, the Bernoulli bond percolation process with parameter $1-p(T_1)$.
Of course, one may wonder whether the thermodynamic limits  $\langle \cdot \rangle^{T_2}_{\sigma^{D(T_1),+}}$ exist.
We are not going to prove it; what is said below holds for any limit point of that family. Note that \eqref{es} contains these states $\nu(T,T)$ with $T_1=T_2=T$ -- and that is why it is useful to speak about the temperature $T_1$, but of course the relevant parameter is the density $p(T_1)$.  The following is therefore presented in the $(p,T)$-plane, which is also the setting of \cite{CCCST}.\\

Consider the curve
\[
  \mathcal{T}_{\operatorname{SG}} (p) = \max\{\frac{1}{\operatorname{arctanh} [\frac{1}{k (1 - 2
  p)}]},0\}.
\]
Note that $\mathcal{T}_{\operatorname{SG}} (p)> 0$ when $k(1-2p) >1$.
%{\color{blue}
%See
% Fig.\    \ref{twoTplot}
% for a picture of the curve}.

\begin{figure}[!h]
\centering
\includegraphics[width=8cm]{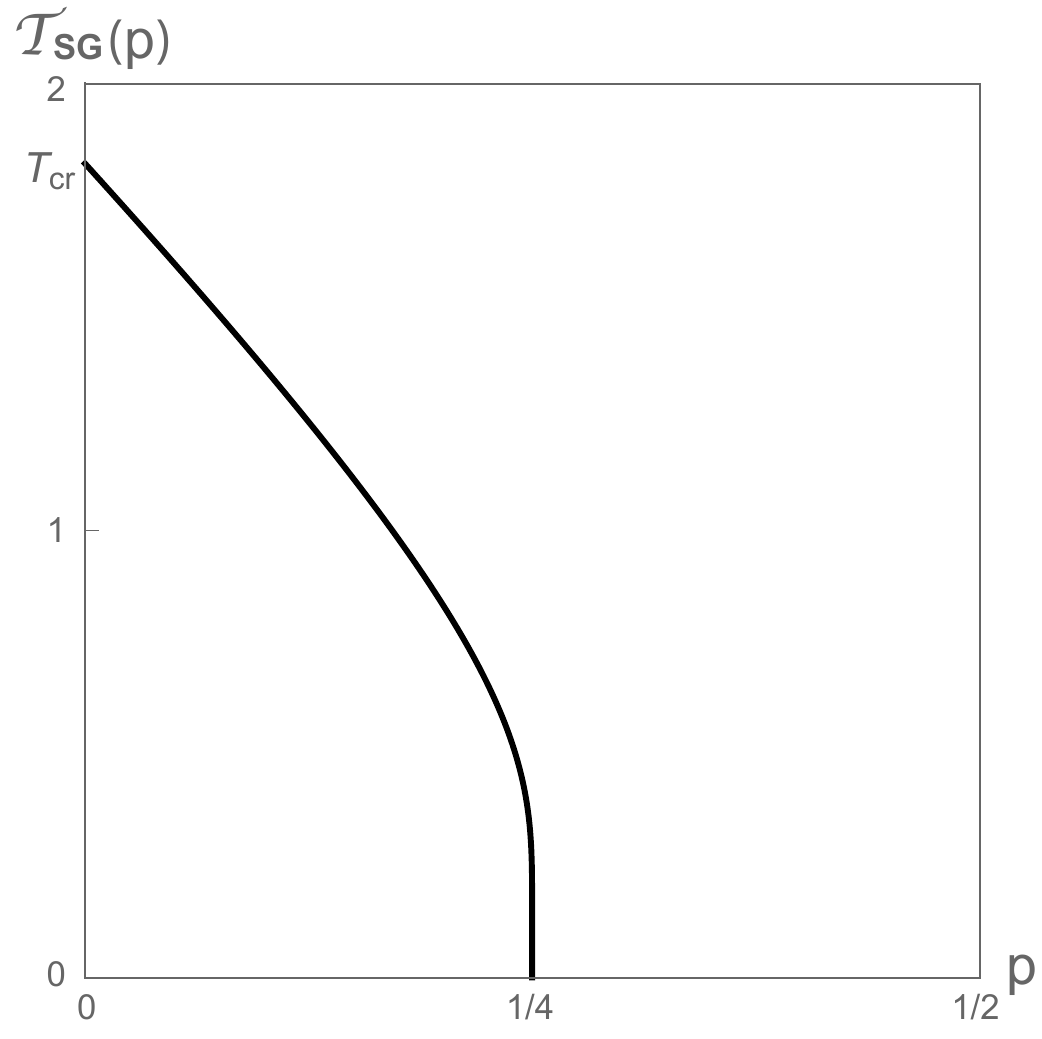}
\caption{The curve $\mathcal{T}_{\operatorname{SG}} (p)$.}
\label{twoTplot}
\end{figure}

\begin{proposition}
\label{magnetisation}
  For any positive temperature $T > 0$ and  parameter $0 \leq p \leq 1/2$,
  \begin{enumerate}
    \item When $T \geq \mathcal{T}_{\operatorname{SG}} (p)$, the expected local
    magnetization of the random Gibbs states
   $
  \langle \cdot \rangle^T_{\sigma^{D, \omega}}
$
 vanishes,
    \[
     \mathbb{E}_p
      (\langle \sigma_x \rangle^T_{\sigma^{D, \pm}}) = 0
    \]
    \item When $T <\mathcal{T}_{\operatorname{SG}} (p)$,
    \[
     \mathbb{E}_p (\langle \sigma_x \rangle^T_{\sigma^{D, +}}) = - \mathbb{E}_p (\langle \sigma_x \rangle^T_{\sigma^{D, -}}) > 0 \]
      \end{enumerate}
\end{proposition}

Let us now look at the second moment, $
  \mathbb{E}_p ([\langle \sigma_x \rangle^T_{\sigma^{D, +}}]^2)=\mathbb{E}_p ([\langle \sigma_x \rangle^T_{\sigma^{D, -}}]^2)$,
which is called the Edwards-Anderson (EA) parameter.

\begin{proposition}
\label{EA}
  For the random Gibbs state $\langle \cdot \rangle^T_{\sigma^{D, +}}$
  \begin{enumerate}
    \item If $T \geq T_{\operatorname{SG}}$ and $T \geq \mathcal{T}_{\operatorname{SG}} (p)$,
    then
    \[
      \mathbb{E}_p ([\langle \sigma_x \rangle^T_{\sigma^{D,
      +}}]^2) = 0
    \]
    \item Otherwise, for any {other} temperature $T > 0$:
    \[
      \mathbb{E}_p ([\langle \sigma_x \rangle^T_{\sigma^{D,
      +}}]^2) > 0
    \]
    Moreover,
      \[
      \textrm{Var}_p (\langle \sigma_x \rangle^T_{\sigma^{D,
      +}}) > 0,
    \]
    which means that the EA random variable $\langle \sigma_x \rangle^T_{\sigma^{D, +}}$ is non-degenerate.
  \end{enumerate}

\end{proposition}

\begin{figure}[!h]
\centering
\includegraphics[width=9cm]{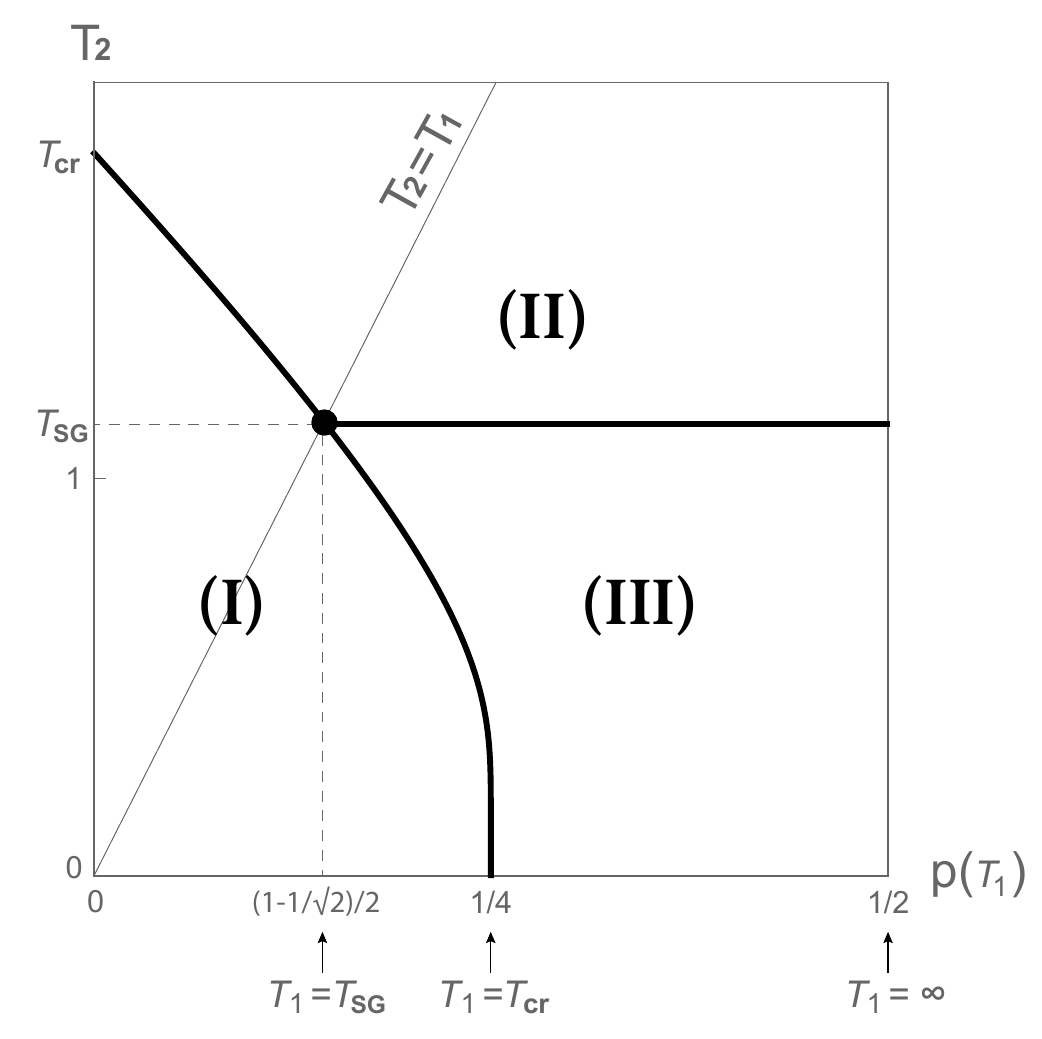}
\caption{Phase diagram. Phases I, II, III correspond to the behavior in Proposition IV.1, IV.2.  See also \cite{CCCST1,CCCST2} for a more qualitative discussion.}
\label{T2(T1)}
\end{figure}

From Fig.~\ref{T2(T1)} it is clear that the model shows an interesting and non-trivial behavior even on the line
 $$
T_1=\infty.
$$
That case is treated in \cite{B2}. \\
 To understand the nature of the spin-glass temperature, we remark that the composition $\mathcal{T}_{\operatorname{SG}} (p(T))$, which we abbreviate as
   \[
\mathcal{T}_{\operatorname{SG}} (T) := \mathcal{T}_{\operatorname{SG}} (p(T))        = \frac{1}{\operatorname{arctanh} \left( \frac{1}{k
      \tanh (1 / T)} \right)}
    \]
   is an involution: $\mathcal{T}_{\operatorname{SG}}(\mathcal{T}_{\operatorname{SG}} (T)) = T$. In particular,
  \[ \mathcal{T}_{\operatorname{SG}}(0) = T_{\operatorname{cr}}, \hspace{1em} \mathcal{T}_{\operatorname{SG}}(T_{\operatorname{cr}}) = 0,
     \hspace{1em} \mathcal{T}_{\operatorname{SG}}(T_{\operatorname{SG}}) = T_{\operatorname{SG}}
     %, \hspace{1em} \mathcal{T}_{\operatorname{SG}}
 %    (\mathcal{T}_{\operatorname{SG}} (T)) = T,
 .\]

Let us prove the last two Propositions in the vicinity of the point
$T_{1}=T_{2}=0.$ After the analysis of the section III.B and using its
technique, it is almost immediate.

Let us fix $x$ to be the root $\mathbf{0}$ of our tree. Informally speaking,
the magnetization at the root $\mathbf{0}$ in the state $\left\langle
\ast\right\rangle _{\sigma^{D,+}}^{T_{2}}$ is defined by the few bonds of the
(rare) bond configuration $D,$ which are in some proximity to $\mathbf{0.}$
Moreover, this magnetization will take different values when these few bonds
happen to be different. Since that happens with positive probability, our
claim follows.

To be more formal, let $B_{R}$ be the ball of radius $R$ centered at
$\mathbf{0.}$ Let $b$ be a bond in $B_{R}.$ Define the set $\mathcal{D}_{b}$
as the family of all realisations $D$ which has $b$ as the only bond in
$B_{R}.$ Clearly, the probability $\mathbf{\Pr}\left(  \mathcal{D}_{b}\right)
$ is positive for every value of the parameter $p\left(  T_{1}\right)  .$ Now,
let $b^{\prime},b^{\prime\prime}$ be two such bonds, with $\mathrm{dist}%
\left(  b^{\prime},\mathbf{0}\right)  >\mathrm{dist}\left(  b^{\prime\prime
},\mathbf{0}\right)  ,$ and let $D^{\prime}\in\mathcal{D}_{b^{\prime}},$
$D^{\prime\prime}\in\mathcal{D}_{b^{\prime\prime}}$ be two typical
configurations. Then, using the technique of the section III.B and a little of
cluster expansions, one sees that there exists a constant $c=c\left(
b^{\prime},b^{\prime\prime}\right)  ,$ such that%
\[
\left\langle \sigma_{\mathbf{0}}\right\rangle _{\sigma^{D^{\prime},+}}^{T_{2}%
}>c>\left\langle \sigma_{\mathbf{0}}\right\rangle _{\sigma^{D^{\prime\prime
},+}}^{T_{2}},
\]
provided both $T_{1}$ and $T_{2}$ are small enough. That proves the positivity
of the variance of the EA random variable $\left\langle \sigma_{\mathbf{0}%
}\right\rangle _{\sigma^{D,+}}^{T}$ (with randomness coming from $D$).

The proofs in general case involve the recursion relations, and can be
obtained from the results of the papers \cite{CCCST1}, \cite{CCCST2}. These results
are summarized graphically by the following pictures.

\newpage

\begin{figure}[h]
\centering
\includegraphics[width=10cm]{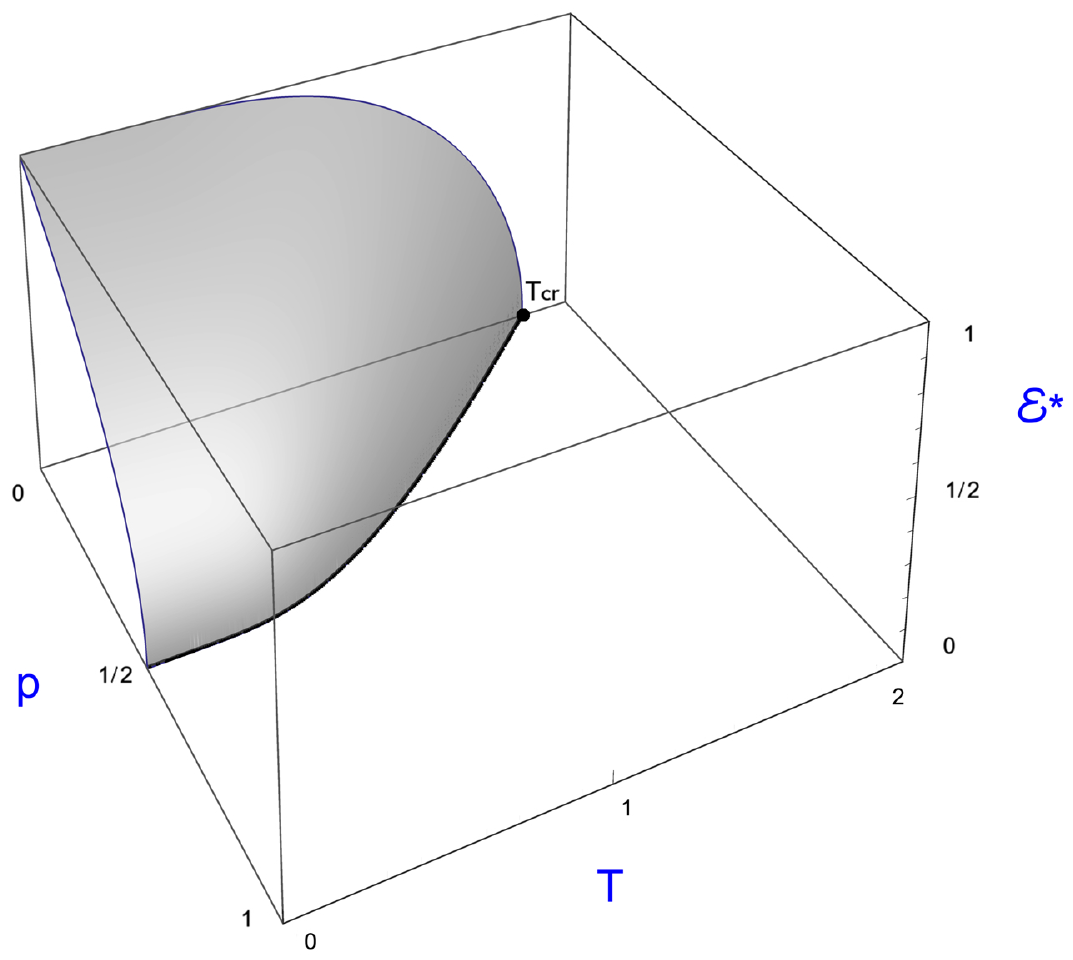}\caption{The first moment.}%
\label{firstmomentgraph}%
\end{figure}

\newpage

\begin{figure}[h]
\centering
\includegraphics[width=10cm]{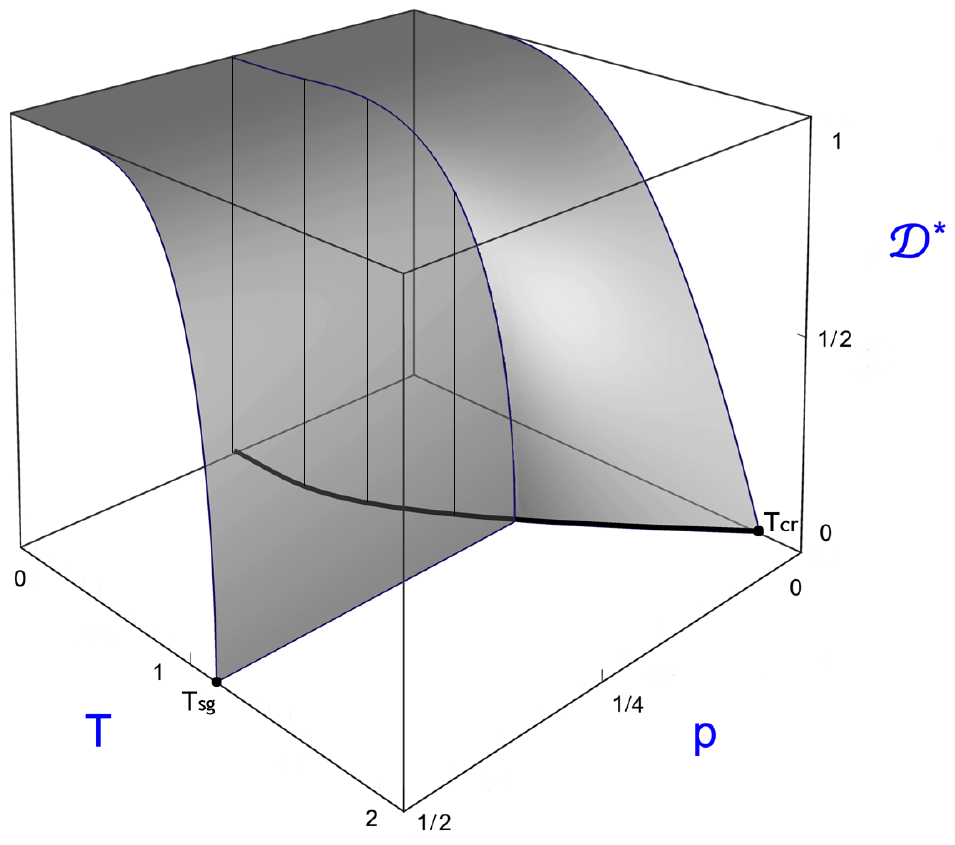}\caption{The second moment.}%
\label{secdmomentgraph}%
\end{figure}

%\begin{figure}
%\begin{minipage}[c]{.46\linewidth}
%\includegraphics[width=10cm]{FigureDm.pdf}
%\end{minipage} \hfill
%\begin{minipage}[c]{.2
%\linewidth}
%\includegraphics[width=4cm]{recursionDM-top.pdf}
%\end{minipage}
%\end{figure}

\newpage
\bigskip

\noindent\textbf{Acknowledgments} Part of this work has been carried out in
the framework of the Labex Archimede (ANR-11-LABX-0033) and of the A*MIDEX
project (ANR-11-IDEX-0001-02), funded by the \textquotedblleft Investissements
d'Avenir" French Government programme managed by the French National Research
Agency (ANR). Part of this work, concerning the EA parameter, has been carried
out at IITP RAS. The support of Russian Foundation for Sciences (project No.
14-50-00150) is gratefully acknowledged. This work was partially supported by
the CNRS PICS grant ``Interfaces al\'eatoires discr\`etes et dynamiques de
Glauber'' and by the grant PRC No. 1556 CNRS-RFBR 2017-2019
`Multi-dimensional semi-classical problems of
Condensed Matter Physics and Quantum Mechanics ''. CM thanks the hospitality of the CPT-Luminy at Marseille.


\begin{thebibliography}{HH}
\bibitem{Io}
  D. Ioffe, A note on the extremality of the disordered state for the Ising model on the Bethe lattice.
  \textit{Lett. Math. Phys.} \textbf{37}: 137-143 (1996).
\end{thebibliography}
\end{document}